\documentclass[11pt]{article}

\usepackage[margin=1in]{geometry}
\usepackage[pdftex]{graphicx}
\usepackage[update,prepend]{epstopdf}
\usepackage[justification=justified,singlelinecheck=false,font=small,labelfont=bf]{caption}
\usepackage{subcaption}
\usepackage{booktabs} 
\usepackage{tabularx}
\newcolumntype{Y}{>{\centering\arraybackslash}X} 
\usepackage{amsmath}
\allowdisplaybreaks[4]
\usepackage{amssymb}
\usepackage{amsfonts}
\usepackage{amsthm}
\usepackage{dsfont}
\usepackage{bold-extra}
\usepackage{titlesec}
\usepackage[title]{appendix}
\usepackage{cases}
\usepackage{enumerate}
\usepackage{verbatim}
\usepackage{xcolor}
\usepackage{mathtools}

\DeclarePairedDelimiter\floor{\lfloor}{\rfloor}
\allowdisplaybreaks[4]
\usepackage{silence} 
\usepackage{hyperref}
\hypersetup{pdfpagemode=UseNone} 
\hypersetup{pdfstartview={XYZ null null 0.85}} 

\newcommand{\RR}{\mathbb{R}}

\DeclareMathOperator{\E}{\mathbb{E}}
\DeclareMathOperator{\Prob}{\mathbb{P}}
\DeclareMathOperator{\Ind}{\mathds{1}}
\DeclareMathOperator{\abs}{\text{abs}_{\infty}}
\DeclareMathOperator{\sgn}{\text{sgn}}

\newtheorem{theorem}{Theorem}[section]
\newtheorem{proposition}[theorem]{Proposition}
\newtheorem{lemma}[theorem]{Lemma}

\numberwithin{equation}{section}

\title{Strong order 1/2 convergence of full truncation Euler\\ approximations to the Cox--Ingersoll--Ross process}

\author{\scshape{andrei cozma}\thanks{\footnotesize\scshape{Mathematical Institute, University of Oxford, Oxford OX2 6GG, United Kingdom}\protect\\ \footnotesize\scshape{The first author gratefully acknowledges financial support from the Engineering and Physical Sciences Research Council (research grant EPSRC EP/N509711/1).}} \and \scshape{christoph reisinger\footnotemark[1]}}

\date{}
\begin{document}
\maketitle

\begin{abstract}
\noindent
We study convergence properties of the full truncation Euler scheme for the Cox--Ingersoll--Ross process in the regime where the boundary point zero is inaccessible. Under some conditions on the model parameters (precisely, when the Feller ratio is greater than three), we establish the strong order 1/2 convergence in $L^{p}$ of the scheme to the exact solution. This is consistent with the optimal rate of strong convergence for approximations to the Cox--Ingersoll--Ross process based on sequential evaluations of the driving Brownian motion.

\vspace{1em}

\vspace{.6em}
\end{abstract}

\section{Introduction}\label{sec:intro}

Let $\big(\Omega,\mathcal{F},(\mathcal{F}_{t})_{t\geq0},\mathbb{P}\big)$ be a filtered probability space and let $W^{v}=\big(W^{v}_{t}\big)_{t\geq0}$ be a 1-dimensional~$\mathcal{F}_{t}$-adapted Brownian motion. A Cox--Ingersoll--Ross (CIR) process is defined by the stochastic differential equation (SDE):
\begin{equation}\label{eq1.1}
dv_{t} = k(\theta-v_{t})dt + \xi\sqrt{v_{t}}\,dW^{v}_{t},
\end{equation}
where $v_{0}$, $k$, $\theta$ and $\xi$ are strictly positive real numbers. The SDE admits a unique strong solution, which is strictly positive when $2k\theta\geq\xi^{2}$ by the Feller test \cite{Karatzas:1991}. The CIR process was originally introduced in finance to model the short-term interest rate \cite{Cox:1985}. Nowadays, due to its desirable properties, like non-negativity, mean-reversion and analytical tractability, it plays a key role in the field of option pricing, for instance when modeling squared volatilities in the Heston model \cite{Heston:1993}. For a given $t\geq0$, the CIR process $v_{t}$ has a noncentral chi-squared conditional distribution and its increments can be simulated exactly \cite{Broadie:2006}. However, when pricing financial derivatives written on an underlying process $S=(S_{t})_{t\in[0,T]}$ modeled by a $d$-dimensional SDE, with CIR dynamics in one or more components, we need to evaluate
\begin{equation}\label{eq1.2}
\E\left[f(\hspace{.25pt}S\hspace{.75pt})\right],
\end{equation}
where $f:\mathcal{C}([0,T],\RR^{d})\rightarrow\RR$ is the discounted payoff. This expectation is rarely available in closed form, nor can $S$ be sampled exactly. In this case, we employ Monte Carlo simulation methods \cite{Glasserman:2004} and approximate the solution to the SDE using a suitable discretization scheme. Due to the non-zero probability of the approximation process becoming negative, the standard Euler--Maruyama scheme applied to \eqref{eq1.1} is not well-defined. Possible remedies are to set the process equal to zero when it turns negative (absorption fix), e.g., the full truncation Euler (FTE) scheme studied in this paper, or reflect it in the origin (reflection fix). An overview of the explicit Euler schemes with different fixes at the boundary can be found in \cite{Lord:2010}. Alternatively, we can use an implicit Euler or a Milstein scheme to discretize the CIR process.

Weak convergence is important when estimating expectations of payoffs such as the one in \eqref{eq1.2}. However, strong convergence plays a crucial role in multilevel Monte Carlo methods \cite{Giles:2008,Giles:2009,Kebaier:2005} and may be required for some complex path-dependent derivatives. Furthermore, pathwise convergence follows automatically \cite{Kloeden:2007}.

The classical convergence theory \cite{Higham:2002,Kloeden:1999} does not apply to the CIR process because the square-root diffusion coefficient is not Lipschitz. Consequently, a considerable amount of research has been devoted to the numerical approximation of \eqref{eq1.1} and alternative approaches have been employed to prove the strong convergence of various discretizations for the CIR process.

Results in the literature concerned with positive strong convergence rates for numerical approximations to the CIR process were restricted to the regime where the boundary point is inaccessible until just recently. Strong convergence, at best with a logarithmic rate, of different Euler schemes including the partial truncation, the full truncation, the reflection and the symmetrized Euler schemes was established in \cite{Alfonsi:2005,Deelstra:1998,Gyongy:2011,Higham:2005,Lord:2010}. The first non-logarithmic rate was obtained in \cite{Berkaoui:2008}, where it was shown that the symmetrized Euler scheme converges strongly with the standard order 1/2 to the exact solution, although in a very restricted parameter regime. Strong convergence with order 1/2 of the backward (drift-implicit) Euler--Maruyama (BEM) scheme was later established in \cite{Dereich:2012}, and this rate was improved to 1 in \cite{Alfonsi:2013,Neuenkirch:2014}. Recently, \cite{Bossy:2016} proved the strong convergence with order 1 of the symmetrized Milstein scheme under some restrictive conditions on the parameters, whereas \cite{Chassagneux:2016} proved the strong convergence of a modified Euler--Maruyama scheme with an order between 1/6 and 1 that depends on the parameters.

In the last few years, there has been some development in the accessible boundary case and polynomial rates of strong convergence for an order of up to 1/2 were established in \cite{Hefter:2016} and \cite{Hutzenthaler:2014} for the truncated Milstein scheme and the BEM scheme, respectively.

In this paper, we study the full truncation Euler scheme proposed in \cite{Lord:2010}. This scheme preserves the positivity of the original process, is easy to implement and hence arguably the most widely used scheme in practice. Perhaps most importantly, it is found empirically to produce the smallest bias of all explicit Euler schemes with different fixes at the boundary \cite{Lord:2010}. Consider a uniform grid
\begin{equation}\label{eq2.2.1}
T = N\delta t,\hspace{1em} t_{n}=n\delta t,\hspace{1em} \forall\hspace{1pt} n\in\{0,1,...,N\}.
\end{equation}
We introduce the discrete-time auxiliary process
\begin{equation}\label{eq2.2.2}
\tilde{v}_{t_{n+1}} = \tilde{v}_{t_{n}} + k(\theta-\tilde{v}_{t_{n}}^{+})\delta t + \xi\sqrt{\tilde{v}_{t_{n}}^{+}}\,\delta W^{v}_{t_{n}},\hspace{.75em} \tilde{v}_{0}=v_{0},
\end{equation}
where $v^{+} = \max\left(0,v\right)$ and $\delta W^{v}_{t_{n}} = W^{v}_{t_{n+1}} - W^{v}_{t_{n}}$, its continuous-time interpolation
\begin{equation}\label{eq2.2.3}
\tilde{v}_{t} = \tilde{v}_{t_{n}} + k(\theta-\tilde{v}_{t_{n}}^{+})(t-t_{n}) + \xi\sqrt{\tilde{v}_{t_{n}}^{+}}\big(W^{v}_{t}-W^{v}_{t_{n}}\big)
\end{equation}
as well as the non-negative process
\begin{equation}\label{eq2.2.4}
\bar{v}_{t}=\tilde{v}_{t_{n}}^{+},
\end{equation}
whenever $t \in [t_{n},t_{n+1})$. The convergence in $L^{1}$ of this scheme was proved in \cite{Lord:2010}. The convergence rate, however, remained an open question and our Theorem \ref{Thm3.8} is the first result to address it, to the best of our knowledge. For convenience, define the Feller ratio
\begin{equation}\label{eq3.0.1}
\nu = \frac{2k\theta}{\xi^{2}}\hspace{1pt}.
\end{equation}
We establish the strong convergence in $L^{p}$ with order 1/2 of the scheme in the inaccessible boundary case, specifically, for a Feller ratio above three. Hence, we obtain the optimal strong convergence rate for numerical approximations to the Cox--Ingersoll--Ross process based on $N$ sequential evaluations of the Brownian driver \cite{Hefter:2017b}. The main and novel idea of the proof is to weight the difference between the process $v$ and its approximation $\tilde{v}$ by the former raised to a suitably chosen negative power and prove the strong convergence with a rate of the weighted error. This, in turn, allows us to derive an upper bound for the actual error.

\begin{theorem}\label{Thm3.8}
Suppose that $\nu>3$ and let $2\leq p<\nu-1$. Then the FTE scheme converges strongly in $L^{p}$ with order 1/2, i.e., there exist $N_{0}\in\mathbb{N}$ and a constant $C$ such that, for all $N>N_{0}$,
\begin{equation}\label{eq3.1.57}
\sup_{t\in\left[0,T\right]}\E\big[|v_{t}-\bar{v}_{t}|^{p}\big]^{\frac{1}{p}} \leq CN^{-\frac{1}{2}}.
\end{equation}
\end{theorem}

We mention that the assumption on the Feller ratio from Theorem \ref{Thm3.8}, i.e., $\nu>3$, appears in the literature as a sufficient condition for the strong convergence with a rate of several discretization schemes for the CIR process. For example, this condition ensures in Corollary 4.1 in \cite{Chassagneux:2016} the strong convergence with order 1/2 (and order 1 if $\nu>5$) of the modified Euler--Maruyama scheme, and in Proposition 3.1 in \cite{Neuenkirch:2014} the strong convergence with order 1 of the BEM scheme.

In \cite{Hefter:2017}, a lower error bound was recently derived for all discretization schemes for the CIR process based on equidistant evaluations of the Brownian driver in the regime where the boundary point is accessible. In light of this result, we demonstrate numerically that the FTE scheme achieves an optimal performance -- in the $L^{1}$ sense -- in half of the regime where the boundary point is accessible, where by optimal we mean that the empirical $L^{1}$ convergence rate is the best possible for equidistant discretization schemes for the CIR process.

The remainder of this paper is structured as follows. In Section~\ref{sec:analysis}, we prove the convergence with a rate of the scheme. In Section~\ref{sec:numerics}, we conduct numerical tests for the rate of convergence that validate and complement our theoretical findings. Finally, Section~\ref{sec:conclusion} contains a short discussion.

\section{Convergence analysis}\label{sec:analysis}

We need to control the polynomial moments of the CIR process and its FTE discretization.

\begin{lemma}\label{Lem3.1}
The CIR process has bounded moments, i.e.,
\begin{equation}\label{eq3.0.2}
\sup_{t\in[0,T]}\E\big[v_{t}^{p}\big] < \infty,\hspace{1em} \forall\hspace{.5pt}p>-\nu.
\end{equation}
\end{lemma}
\begin{proof}
Follows from \cite{Dereich:2012} or Theorem 3.1 in \cite{Hurd:2008}.
\end{proof}

\begin{lemma}\label{Lem3.3}
The FTE scheme has uniformly bounded moments, i.e.,
\begin{equation}\label{eq3.1.1}
\sup_{N\geq1}\hspace{1.5pt}\E\bigg[\sup_{t\in[0,T]}|\tilde{v}_{t}|^p\bigg] < \infty,\hspace{1em} \forall\hspace{.5pt}p\geq1.
\end{equation}
\end{lemma}
\begin{proof}
Follows from a simple application of the Burkholder--Davis--Gundy (BDG) inequality and Proposition 3.7 in \cite{Cozma:2016a}.
\end{proof}

By construction, the FTE approximation $\bar{v}$ is non-negative. However, an important step in the convergence analysis lies in analyzing the behaviour of the auxiliary process $\tilde{v}$ at the boundary. The next result derives a polynomial upper bound in the time step size on the probability of $\tilde{v}$ becoming negative. Similar results were established for the symmetrized Euler scheme, in Lemma 3.7 in \cite{Bossy:2007}, and for the symmetrized Milstein scheme, in Lemma 2.2 in \cite{Bossy:2016}. However, the full truncation Euler scheme has led to different technical challenges and the arguments employed in the proofs of the aforementioned results do not apply here.
\begin{proposition}\label{Prop3.5}
Suppose that $\nu>2$ and let
\begin{equation}\label{eq3.1.5}
\bar{\nu} = \inf\left\{x>0:\,\frac{(4x\vee\nu)(\nu-x)}{\nu x(\nu-x-1)} < 1.99\sqrt{\nu\pi}\hspace{1pt}e^{\frac{\nu}{2}}-1\right\}.
\end{equation}
Then there exist $N_{0}\in\mathbb{N}$ and a constant $C$ such that, for all $N>N_{0}$,
\begin{equation}\label{eq3.1.6}
\sup_{0\leq n\leq N}\Prob\big(\tilde{v}_{t_{n}}\leq0\big) \leq CN^{-\nu+\bar{\nu}+1}.
\end{equation}
\end{proposition}
\begin{proof}
We first show that $\bar{\nu}$ exists and derive bounds which imply that the exponent in \eqref{eq3.1.6} is negative. Note that as the value of $\nu$ increases, the left-hand side term of the inequality in \eqref{eq3.1.5} decreases and the right-hand side term increases. Hence, $\bar{\nu}$ decreases as $\nu$ increases. In particular, $0<\bar{\nu}<\bar{\nu}|_{\nu=2}\approx0.176$.

Suppose that $N_{0}\geq\floor{kT}$ and fix $N>N_{0}$. Define, for brevity,
\begin{equation}\label{eq3.1.7}
\alpha_{N} = \frac{1}{2}\bigg(1-\frac{kT}{N}\bigg) = \frac{1-k\delta t}{2}\hspace{1pt}.
\end{equation}

First, consider the sequence $(c_{j})_{0\leq j\leq N}$ given by
\begin{equation}\label{eq3.1.8}
c_{0} = \alpha_{N},\hspace{1em} c_{1} = \alpha_{N}-\alpha_{N}^{2} \hspace{1em}\text{ and }\hspace{1em} c_{j+1} = c_{j}^{2} + \alpha_{N} - \alpha_{N}^{2},\hspace{1em} \forall\hspace{.5pt} 1\leq j\leq N-1.
\end{equation}
As $\alpha_{N}\in(0,0.5)$, one can clearly see that $c_{j}\in(0,\alpha_{N})$ for all $1\leq j\leq N$. We will show by induction that
\begin{equation}\label{eq3.1.9}
c_{j} \leq 1 - \alpha_{N} - \frac{\varphi_{\nu}}{j-1+\eta_{\nu}}\hspace{1pt},\hspace{1em} \forall\hspace{.5pt} 1\leq j\leq N,
\end{equation}
where
\begin{equation}\label{eq3.1.10}
\varphi_{\nu}=1-\frac{\bar{\nu}}{\nu} \hspace{1em}\text{ and }\hspace{1em} \eta_{\nu}=\frac{(\nu-\bar{\nu})(4\bar{\nu}\vee\nu)}{\nu\bar{\nu}}\hspace{1pt}.
\end{equation}
Since
\begin{equation}\label{eq3.1.11}
\frac{\varphi_{\nu}}{\eta_{\nu}} = \frac{\bar{\nu}}{4\bar{\nu}\vee\nu} \leq \frac{1}{4} \leq (1-\alpha_{N})^{2},
\end{equation}
\eqref{eq3.1.9} holds when $j=1$. Suppose that \eqref{eq3.1.9} holds for some $1\leq j\leq N-1$, then
\begin{equation}\label{eq3.1.12}
c_{j+1} \leq \alpha_{N}-\alpha_{N}^{2} +\bigg(1 - \alpha_{N} - \frac{\varphi_{\nu}}{j-1+\eta_{\nu}}\bigg)^{2}
\end{equation}
and some simple computations lead to the following sufficient condition for the induction step,
\begin{equation}\label{eq3.1.13}
(j-1+\eta_{\nu})^{2}(1-2\alpha_{N}) + (j-1+\eta_{\nu})(1-2\alpha_{N}) + (j-1)(1-\varphi_{\nu}) + \eta_{\nu}(1-\varphi_{\nu}) - \varphi_{\nu} \geq 0,
\end{equation}
which clearly holds. For convenience, define another sequence $(a_{j})_{0\leq j\leq N}$ given by
\begin{equation}\label{eq3.1.15}
a_{j} = \frac{2(\alpha_{N}-c_{j})}{\xi^{2}\delta t}\hspace{1pt},\hspace{1em} \forall\hspace{1pt} 0\leq j\leq N,
\end{equation}
such that
\begin{equation}\label{eq3.1.16}
a_{0} = 0,\hspace{1em} a_{1} = \frac{2\alpha_{N}^{2}}{\xi^{2}\delta t} \hspace{1em}\text{ and }\hspace{1em} a_{j+1} = 2\alpha_{N}a_{j} - \frac{1}{2}\hspace{1pt}a_{j}^{2}\xi^{2}\delta t,\hspace{1em} \forall\hspace{.5pt} 1\leq j\leq N-1.
\end{equation}
A sequence similar to $(a_{j})_{0\leq j\leq N}$ was analyzed in \cite{Bossy:2007}. However, a sharper lower bound than the one obtained in Lemma 3.6 in \cite{Bossy:2007} is needed for our purposes. We now show that, for $N$ large enough,
\begin{equation}\label{eq3.1.17}
\mathcal{S}_{1} = \sum_{n=0}^{N-2}{\prod_{j=0}^{n}{\exp\left\{-a_{j}k\theta\delta t \right\}}} \leq 2\sqrt{\nu(1-k\delta t)\pi}\hspace{1pt}e^{\frac{\nu}{2}(1-k\delta t)},
\end{equation}
a bound which will be of use later in the proof. Using \eqref{eq3.1.9}, we get
\begin{align}\label{eq3.1.18}
\mathcal{S}_{1} &= \sum_{n=0}^{N-2}{\exp\Bigg\{\nu\sum_{j=1}^{n}{(c_{j}-\alpha_{N})}\Bigg\}} \nonumber\\[2pt]
&\leq \sum_{n=0}^{N-2}{\exp\Bigg\{\nu\sum_{j=1}^{n}{\bigg(1 - 2\alpha_{N} - \frac{\varphi_{\nu}}{j-1+\eta_{\nu}}\bigg)}\Bigg\}} \nonumber\\[2pt]
&\leq \sum_{n=0}^{N-2}{\exp\bigg\{\nu kT\frac{n}{N}-\nu\varphi_{\nu}\int_{0}^{n}{(x+\eta_{\nu})^{-1}\hspace{1pt}dx}\bigg\}} \nonumber\\[2pt]
&\leq \eta_{\nu}^{\nu\varphi_{\nu}}\sum_{n=0}^{N-1}{e^{\nu kT\frac{n}{N}}(\eta_{\nu}+n)^{-\nu\varphi_{\nu}}}.
\end{align}
Let $\epsilon=0.002$ for the rest of the proof. Since $\nu\varphi_{\nu}=\nu-\bar{\nu}>1$, the Hurwitz (generalized Riemann) zeta function
\begin{equation}\label{eq3.1.19}
\zeta(\nu\varphi_{\nu},\eta_{\nu}) = \sum_{n=0}^{\infty}{(\eta_{\nu}+n)^{-\nu\varphi_{\nu}}}
\end{equation}
converges and hence, for $N$ large enough, we have
\begin{equation}\label{eq3.1.20}
\sum_{n=1}^{N-1}{\big(e^{\nu kT\frac{n}{N}}-1-\epsilon\big)(\eta_{\nu}+n)^{-\nu\varphi_{\nu}}} \leq
\sum_{n>\frac{\log(1+\epsilon)}{\nu kT}N_{0}}{\big(e^{\nu kT}-1\big)(\eta_{\nu}+n)^{-\nu\varphi_{\nu}}} \leq
\epsilon\eta_{\nu}^{-\nu\varphi_{\nu}},
\end{equation}
which implies that
\begin{equation}\label{eq3.1.21}
\sum_{n=0}^{N-1}{e^{\nu kT\frac{n}{N}}(\eta_{\nu}+n)^{-\nu\varphi_{\nu}}} \leq (1+\epsilon)\sum_{n=0}^{N-1}{(\eta_{\nu}+n)^{-\nu\varphi_{\nu}}}
\leq (1+\epsilon)\zeta(\nu\varphi_{\nu},\eta_{\nu}).
\end{equation}
However,
\begin{equation}\label{eq3.1.22}
\zeta(\nu\varphi_{\nu},\eta_{\nu}) = \eta_{\nu}^{-\nu\varphi_{\nu}} + \sum_{n=1}^{\infty}{(\eta_{\nu}+n)^{-\nu\varphi_{\nu}}}
\leq \eta_{\nu}^{-\nu\varphi_{\nu}} + \int_{0}^{\infty}{(\eta_{\nu}+x)^{-\nu\varphi_{\nu}}\hspace{1pt}dx},
\end{equation}
and hence
\begin{equation}\label{eq3.1.23}
\zeta(\nu\varphi_{\nu},\eta_{\nu}) \leq \eta_{\nu}^{-\nu\varphi_{\nu}} + \frac{\eta_{\nu}^{-\nu\varphi_{\nu}+1}}{\nu\varphi_{\nu}-1}\hspace{1pt}.
\end{equation}
Combining \eqref{eq3.1.18}, \eqref{eq3.1.21} and \eqref{eq3.1.23} and using \eqref{eq3.1.5} and \eqref{eq3.1.10}, we deduce that
\begin{equation}\label{eq3.1.24}
\mathcal{S}_{1} \leq 1.99(1+\epsilon)\sqrt{\nu\pi}\hspace{1pt}e^{\frac{\nu}{2}}.
\end{equation}
However, for $N$ large enough, we have
\begin{equation}\label{eq3.1.25}
\sqrt{\nu\pi}\hspace{1pt}e^{\frac{\nu}{2}} \leq (1+\epsilon)\sqrt{\nu(1-k\delta t)\pi}\hspace{1pt}e^{\frac{\nu}{2}(1-k\delta t)}
\end{equation}
and the upper bound on $\mathcal{S}_{1}$ in \eqref{eq3.1.17} follows.

Second, recall from \eqref{eq2.2.2} that, for all $0\leq j\leq N-1$,
\begin{equation}\label{eq3.1.26}
\tilde{v}_{t_{N-j}} = \tilde{v}_{t_{N-j-1}} + k\theta\delta t - k\tilde{v}_{t_{N-j-1}}^{+}\delta t + \xi\sqrt{\tilde{v}_{t_{N-j-1}}^{+}}\big(W^{v}_{t_{N-j}} - W^{v}_{t_{N-j-1}}\big),
\end{equation}
and note that
\begin{equation}\label{eq3.1.27}
\Prob\big(\tilde{v}_{t_{N-j}}\leq0\big) = \Prob\big(\tilde{v}_{t_{N-j-1}}\leq-k\theta\delta t\big) + \Prob\big(\tilde{v}_{t_{N-j}}\leq0,\tilde{v}_{t_{N-j-1}}>0\big).
\end{equation}
Let $(\mathcal{F}^{v}_{t})_{t\geq0}$ be the natural filtration generated by $W^{v}$ and consider the shorthand notations $\E_{t}\!\big[\hspace{.5pt}\cdot\hspace{.5pt}\big]=\E\big[\hspace{.5pt}\cdot\hspace{.5pt}|\hspace{1pt}\mathcal{F}^{v}_{t}\big]$ and $\Prob_{t}\!\big(\hspace{.5pt}\cdot\hspace{.5pt}\big)=\Prob\big(\hspace{.5pt}\cdot\hspace{.5pt}|\hspace{1pt}\mathcal{F}^{v}_{t}\big)$ for the conditional expectation and probability. Conditioning on $\mathcal{F}^{v}_{t_{N-j-1}}$, we get
\begin{align}\label{eq3.1.28}
\Prob\big(\tilde{v}_{t_{N-j}}\leq0,\tilde{v}_{t_{N-j-1}}>0\big) &= \E\Big[\Ind_{\tilde{v}_{t_{N-j-1}}>0}\hspace{1pt}\E_{t_{N-j-1}}\!\Big[\Ind_{\tilde{v}_{t_{N-j}}\leq0}\Big]\Big] \nonumber\\[3pt]
&= \E\bigg[\Ind_{w>0}\hspace{1pt}\Prob_{t_{N-j-1}}\!\bigg(Z\leq-\frac{k\theta\delta t+w^{+}(1-k\delta t)}{\xi\sqrt{w^{+}\delta t}}\bigg)\bigg],
\end{align}
where $w = \tilde{v}_{t_{N-j-1}}$ and $Z\sim\mathcal{N}\left(0,1\right)$ independent of $\mathcal{F}^{v}_{t_{N-j-1}}$. Using a standard inequality for the lower tail of the normal distribution, namely
\begin{equation}\label{eq3.1.29}
\Prob\big(Z\leq-x\big) \leq \frac{1}{\sqrt{2\pi}\hspace{1pt}x}\hspace{1pt}e^{-\frac{1}{2}x^{2}},\hspace{1em} \forall x>0,
\end{equation}
and the arithmetic mean-geometric mean (AM-GM) inequality, we deduce that
\begin{align}\label{eq3.1.30}
\Prob\big(\tilde{v}_{t_{N-j}}\leq0,\tilde{v}_{t_{N-j-1}}>0\big) &\leq
\frac{1}{2\sqrt{\nu(1-k\delta t)\pi}}\hspace{1pt}\E\bigg[\Ind_{w>0}\hspace{1pt}\exp\bigg\{-\frac{(k\theta\delta t+w^{+}(1-k\delta t))^{2}}{2\xi^{2}w^{+}\delta t}\bigg\}\bigg] \nonumber\\[2pt]
&\leq \frac{e^{-\frac{\nu}{2}(1-k\delta t)}}{2\sqrt{\nu(1-k\delta t)\pi}}\hspace{1pt}\E\Big[\exp\left\{-a_{1}\abs(\tilde{v}_{t_{N-j-1}})\right\}\Big],
\end{align}
where $\abs(w)=w$ if $w>0$ and $\abs(w)=\infty$ otherwise. Let $1\leq i,j\leq N-1$, then
\begin{align}\label{eq3.1.31}
\E\Big[\exp\left\{-a_{i}\abs(\tilde{v}_{t_{N-j}})\right\}\Big] &= \E\Big[\exp\left\{-a_{i}\abs(\tilde{v}_{t_{N-j}})\right\}\left(\Ind_{\tilde{v}_{t_{N-j-1}}>0}+\Ind_{\tilde{v}_{t_{N-j-1}}\leq0}\right)\Big] \nonumber\\[2pt]
&\leq \E\Big[\Ind_{\tilde{v}_{t_{N-j-1}}>0}\hspace{1pt}\E_{t_{N-j-1}}\!\Big[\exp\left\{-a_{i}\abs(\tilde{v}_{t_{N-j}})\right\}\Big]\Big] \nonumber\\[3pt]
&+ \Prob\big(-k\theta\delta t<\tilde{v}_{t_{N-j-1}}\leq0\big).
\end{align}
Denote $w = \tilde{v}_{t_{N-j-1}}$ as before and let $\mathcal{I}$ be the inner expectation on the right-hand side of \eqref{eq3.1.31}, i.e.,
\begin{equation}\label{eq3.1.32}
\mathcal{I} = \E_{t_{N-j-1}}\!\Big[\exp\left\{-a_{i}\abs\!\big(w + k\theta\delta t - kw^{+}\delta t + \xi\sqrt{w^{+}\delta t}\hspace{1pt}Z\big)\right\}\Big],
\end{equation}
where $Z\sim\mathcal{N}\left(0,1\right)$ independent of $\mathcal{F}^{v}_{t_{N-j-1}}$. There are two possible outcomes, namely $w\leq0$, in which case $\mathcal{I}\leq1$, and $w>0$, which is treated now:
\begin{align}\label{eq3.1.33}
\mathcal{I} &= \int_{z_{0}}^{\infty}{\frac{1}{\sqrt{2\pi}}\exp\left\{-\frac{1}{2}z^{2} -a_{i}\big(k\theta\delta t + w(1-k\delta t) + \xi\sqrt{w\delta t}\hspace{1pt}z\big)\right\}dz} \nonumber\\[2pt]
&= \exp\left\{-a_{i}k\theta\delta t - w a_{i}(1-k\delta t) + \frac{1}{2}\hspace{1pt}w a_{i}^{2}\xi^{2}\delta t\right\}\Phi\left(-z_{0}-a_{i}\xi\sqrt{w\delta t}\hspace{1pt}\right),
\end{align}
where
\begin{equation}\label{eq3.1.34}
z_{0} = -\frac{k\theta\delta t + w(1-k\delta t)}{\xi\sqrt{w\delta t}}
\end{equation}
and $\Phi$ is the standard normal CDF. We deduce from \eqref{eq3.1.16} that
\begin{equation}\label{eq3.1.35}
\mathcal{I} \leq \exp\left\{-a_{i}k\theta\delta t - w a_{i+1}\right\},
\end{equation}
and hence that
\begin{align}\label{eq3.1.36}
\E\Big[\exp\left\{-a_{i}\abs(\tilde{v}_{t_{N-j}})\right\}\Big] &\leq \exp\left\{-a_{i}k\theta\delta t\right\}\E\Big[\exp\left\{-a_{i+1}\abs(\tilde{v}_{t_{N-j-1}})\right\}\Big] \nonumber\\[2pt]
&+ \Prob\big(-k\theta\delta t<\tilde{v}_{t_{N-j-1}}\leq0\big).
\end{align}
For $N$ large enough, putting together \eqref{eq3.1.17}, \eqref{eq3.1.27}, \eqref{eq3.1.30} and \eqref{eq3.1.36}, we can prove by induction that, for all $0\leq l\leq N-1$,
\begin{align}\label{eq3.1.37}
\sup_{0\leq n\leq l}\Prob\big(\tilde{v}_{t_{N-n}}\leq0\big) &\leq \frac{e^{-\frac{\nu}{2}(1-k\delta t)}}{2\sqrt{\nu(1-k\delta t)\pi}}\sum_{n=0}^{l}{\E\Big[\exp\left\{-a_{n+1}\abs(\tilde{v}_{t_{N-l-1}})\right\}\Big]\prod_{j=0}^{n}{\exp\left\{-a_{j}k\theta\delta t\right\}}} \nonumber\\[2pt]
&+ \Prob\big(\tilde{v}_{t_{N-l-1}}\leq0\big).
\end{align}
Taking $l=N-1$ in \eqref{eq3.1.37} and since $\tilde{v}_{t_{0}}=v_{0}>0$, we obtain
\begin{equation}\label{eq3.1.38}
\sup_{0\leq n\leq N}\Prob\big(\tilde{v}_{t_{n}}\leq0\big) \leq \frac{e^{-\frac{\nu}{2}(1-k\delta t)}}{2\sqrt{\nu(1-k\delta t)\pi}}\sum_{n=0}^{N-1}{\exp\left\{-a_{n+1}v_{0}\right\}\prod_{j=0}^{n}{\exp\left\{-a_{j}k\theta\delta t\right\}}}.
\end{equation}
Using \eqref{eq3.1.9} yields
\begin{align}\label{eq3.1.39}
\mathcal{S}_{2} &= \sum_{n=0}^{N-1}{\exp\left\{-a_{n+1}v_{0}\right\}\prod_{j=0}^{n}{\exp\left\{-a_{j}k\theta\delta t\right\}}} \nonumber\\[2pt]
&\leq \sum_{n=0}^{N-1}{\exp\Bigg\{\nu\sum_{j=1}^{n}{\bigg(1 - 2\alpha_{N} - \frac{\varphi_{\nu}}{j-1+\eta_{\nu}}\bigg)}+\frac{2v_{0}}{\xi^{2}\delta t}\bigg(1 - 2\alpha_{N} - \frac{\varphi_{\nu}}{n+\eta_{\nu}}\bigg)\Bigg\}} \nonumber\\[2pt]
&\leq \eta_{\nu}^{\nu\varphi_{\nu}}\exp\left\{\nu\bigg(\frac{v_{0}}{\theta}+kT\bigg)\right\}\sum_{n=0}^{N-1}{(\eta_{\nu}+n)^{-\nu\varphi_{\nu}}}\exp\left\{-\frac{2v_{0}\varphi_{\nu}N}{\xi^{2}T(n+\eta_{\nu})}\right\}.
\end{align}
However, for all $x,y>0$, we have that $\log(y/x)\geq1-x/y$ such that $e^{-x}\leq e^{-y}(y/x)^{y}$, and hence
\begin{align}\label{eq3.1.40}
\mathcal{S}_{2} &\leq \eta_{\nu}^{\nu\varphi_{\nu}}\exp\left\{\nu\bigg(\frac{v_{0}}{\theta}+kT\bigg)\right\}\sum_{n=0}^{N-1}{e^{-\nu\varphi_{\nu}}\left(\frac{k\theta T}{v_{0}}\right)^{\nu\varphi_{\nu}}N^{-\nu\varphi_{\nu}}} \nonumber\\[2pt]
&= \left(\frac{k\theta T\eta_{\nu}}{v_{0}}\right)^{\nu\varphi_{\nu}}\exp\left\{\nu\bigg(\frac{v_{0}}{\theta}+kT-\varphi_{\nu}\bigg)\right\}N^{-\nu+\bar{\nu}+1}.
\end{align}
Combining \eqref{eq3.1.25}, \eqref{eq3.1.38} and \eqref{eq3.1.40}, we conclude that there exists $N_{0}\in\mathbb{N}$ such that
\begin{equation}\label{eq3.1.41}
\sup_{0\leq n\leq N}\Prob\big(\tilde{v}_{t_{n}}\leq0\big) \leq \frac{(1+\epsilon)e^{-\frac{\nu}{2}}}{2\sqrt{\nu\pi}}\hspace{1pt}\left(\frac{k\theta T\eta_{\nu}}{v_{0}}\right)^{\nu\varphi_{\nu}}\exp\left\{\nu\bigg(\frac{v_{0}}{\theta}+kT-\varphi_{\nu}\bigg)\right\}N^{-\nu+\bar{\nu}+1}
\end{equation}
for all $N>N_{0}$.
\end{proof}

Next, we bound the $L^{p}$ difference between the two continuous-time approximations $\tilde{v}$ and $\bar{v}$.

\begin{proposition}\label{Prop3.6}
Suppose that $\nu>2$ and let $1\leq p<2(\nu-\bar{\nu}-1)$. Then there exist $N_{0}\in\mathbb{N}$ and a constant $C$ such that, for all $N>N_{0}$,
\begin{equation}\label{eq3.1.42}
\sup_{t\in[0,T]}\E\big[|\tilde{v}_{t}-\bar{v}_{t}|^{p}\big]^{\frac{1}{p}} \leq CN^{-\frac{1}{2}}.
\end{equation}
\end{proposition}
\begin{proof}
Let $\epsilon=2(\nu-\bar{\nu}-1)-p>0$. For convenience of notation, define
\begin{equation}\label{eq3.1.43}
\Delta\tilde{v}_{t} = \tilde{v}_{t}-\tilde{v}_{\bar{t}},
\end{equation}
where $\bar{t}=\delta t\floor*{\frac{t}{\delta t}}$ for all $t\in[0,T]$. From the triangle inequality, we have
\begin{equation}\label{eq3.1.44}
|\tilde{v}_{t}-\bar{v}_{t}| \leq |\tilde{v}_{\bar{t}}-\tilde{v}_{\bar{t}}^{+}| + |\Delta\tilde{v}_{t}| = |\tilde{v}_{\bar{t}}|\Ind_{\tilde{v}_{\bar{t}}\leq0} + |\Delta\tilde{v}_{t}|
\end{equation}
and hence, using H\"older's inequality, we get
\begin{align}\label{eq3.1.45}
\sup_{t\in[0,T]}\E\big[|\tilde{v}_{t}-\bar{v}_{t}|^{p}\big] &\leq 2^{p-1}\sup_{N\geq1}\hspace{1.5pt}\sup_{0\leq n\leq N}\E\Big[|\tilde{v}_{t_{n}}|^{\frac{p(p+\epsilon)}{\epsilon}}\Big]^{\frac{\epsilon}{p+\epsilon}}\sup_{0\leq n\leq N}\Prob\big(\tilde{v}_{t_{n}}\leq0\big)^{\frac{p}{p+\epsilon}} \nonumber\\[2pt]
&+ 2^{p-1}\sup_{t\in[0,T]}\E\big[|\Delta\tilde{v}_{t}|^{p}\big].
\end{align}
We can bound the last term on the right-hand side from above as follows,
\begin{align}\label{eq3.1.46}
|\Delta\tilde{v}_{t}|^{p} &\leq
\Big(k\theta\delta t + k\bar{v}_{t}\delta t + \xi\sqrt{\bar{v}_{t}}\hspace{1pt}\big|W^{v}_{t}-W^{v}_{\bar{t}}\big|\Big)^{p} \nonumber\\[3pt]
&\leq 3^{p-1}k^{p}\theta^{p}(\delta t)^{p} + 3^{p-1}k^{p}|\tilde{v}_{\bar{t}}|^{p}(\delta t)^{p} + 3^{p-1}\xi^{p}|\tilde{v}_{\bar{t}}|^{\frac{p}{2}}\big|W^{v}_{t}-W^{v}_{\bar{t}}\big|^{p},
\end{align}
and hence
\begin{align}\label{eq3.1.47}
\sup_{t\in[0,T]}\E\big[|\Delta\tilde{v}_{t}|^{p}\big] &\leq 3^{p-1}(k\theta T)^{p}N^{-p} + 3^{p-1}(kT)^{p}\sup_{N\geq1}\hspace{1.5pt}\sup_{0\leq n\leq N}\E\big[|\tilde{v}_{t_{n}}|^{p}\big]N^{-p} \nonumber\\[2pt]
&+ 3^{p-1}\xi^{p}T^{\frac{p}{2}}\E\big[|Z|^{p}\big]\sup_{N\geq1}\hspace{1.5pt}\sup_{0\leq n\leq N}\E\Big[|\tilde{v}_{t_{n}}|^{\frac{p}{2}}\Big]N^{-\frac{p}{2}},
\end{align}
where $Z\sim\mathcal{N}\left(0,1\right)$. Substituting back into \eqref{eq3.1.45} with \eqref{eq3.1.47} and using Lemma \ref{Lem3.3} and Proposition \ref{Prop3.5} concludes the proof.
\end{proof}

Before we prove the main result of this paper, we need the following technical auxiliary result.

\begin{lemma}\label{Lem3.7}
Suppose that $2\leq q<\nu-1$. Then we can find $\beta>0$ such that
\begin{equation}\label{eq3.1.48}
\nu > 2\beta+1 > \nu+2q-\sqrt{(\nu+2q-1)^{2}-4q(q-1)}
\end{equation}
and
\begin{equation}\label{eq3.1.49}
2(\nu-\bar{\nu}-1)(\nu-\beta-1) > \nu q.
\end{equation}
\end{lemma}
\begin{proof}
Note that \eqref{eq3.1.48} and \eqref{eq3.1.49} are equivalent to
\begin{equation}\label{eq3.1.50}
(\nu q) \vee (\nu-1)(\nu-\bar{\nu}-1) < (\nu-\bar{\nu}-1)\big(\nu-2q-1+\sqrt{(\nu+2q-1)^{2}-4q(q-1)}\hspace{1pt}\big)
\end{equation}
since $\beta\in(0,\infty)$ spans the interval of values associated with \eqref{eq3.1.50} for
\begin{equation}\label{eq3.1.51}
2(\nu-\bar{\nu}-1)(\nu-\beta-1).
\end{equation}
Furthermore, one can easily check that
\begin{equation}\label{eq3.1.52}
2q < \sqrt{(\nu+2q-1)^{2}-4q(q-1)}\hspace{1pt},
\end{equation}
which implies that
\begin{equation}\label{eq3.1.53}
1+\bar{\nu} < \nu-\frac{\nu q}{\nu-2q-1+\sqrt{(\nu+2q-1)^{2}-4q(q-1)}}
\end{equation}
is equivalent to \eqref{eq3.1.50}. However, we can rewrite \eqref{eq3.1.53} as
\begin{equation}\label{eq3.1.54}
1+\bar{\nu} < \nu\bigg(1-\frac{2q+1-\nu+\sqrt{(\nu+2q-1)^{2}-4q(q-1)}}{4(2\nu-q-1)}\hspace{1pt}\bigg).
\end{equation}
Let $x=\nu-q-1>0$ and define the right-hand side function
\begin{align}\label{eq3.1.55}
f_{q}(x) &= (q+1+x)\bigg(1-\frac{4q-(3q+x)+\sqrt{(3q+x)^{2}-4q(q-1)}}{4(q+1+2x)}\hspace{1pt}\bigg) \nonumber\\[2pt]
&= 1 + \frac{(2q+1+2x)x}{q+1+2x} + \frac{(2q+2+2x)q(q-1)}{2(q+1+2x)(3q+x+\sqrt{(3q+x)^{2}-4q(q-1)}\hspace{1pt})}\hspace{1pt}.
\end{align}
Hence, we deduce that
\begin{align}\label{eq3.1.56}
f_{q}(x) \geq 1 + x + \frac{q(q-1)}{4(3q+x)} \geq 1 + \frac{q-1}{12} \geq 1.083 > 1.057 = 1 + \bar{\nu}|_{\nu=3} \geq 1 + \bar{\nu},
\end{align}
which concludes the proof.
\end{proof}

With these results at our disposal, we are now ready to prove the main theorem.

\begin{proof}(Theorem \ref{Thm3.8})
Let $p<q<\nu-1$ and fix any $\beta>0$ which satisfies \eqref{eq3.1.48} and \eqref{eq3.1.49}. For convenience of notation, define
\begin{equation}\label{eq3.1.58}
e^{v}_{t}=v_{t}-\tilde{v}_{t},\hspace{.75em} e^{v}_{0}=0,
\end{equation}
such that
\begin{equation}\label{eq3.1.59}
de^{v}_{t} = -k(v_{t}-\bar{v}_{t})dt + \xi\big(\sqrt{v_{t}}-\sqrt{\bar{v}_{t}}\hspace{1pt}\big)dW^{v}_{t}.
\end{equation}
Since $\nu\geq1$, the CIR process $v$ has almost surely strictly positive paths. Applying It\^o's formula to the $\mathcal{C}^{2,2}$ function $f(v_{t},e^{v}_{t})=v_{t}^{-\beta}|e^{v}_{t}|^{q}$ yields
\begin{align}\label{eq3.1.60}
v_{t}^{-\beta}|e^{v}_{t}|^{q} &= -\beta\int_{0}^{t}{v_{u}^{-(\beta+1)}|e^{v}_{u}|^{q}\,dv_{u}} + q\int_{0}^{t}{v_{u}^{-\beta}|e^{v}_{u}|^{q-1}\sgn(e^{v}_{u})\,de^{v}_{u}} + \frac{1}{2}\hspace{1pt}\beta(\beta+1)\int_{0}^{t}{v_{u}^{-(\beta+2)}|e^{v}_{u}|^{q}\,d\langle v\rangle_{u}} \nonumber\\[2pt]
&+ \frac{1}{2}\hspace{1pt}q(q-1)\int_{0}^{t}{v_{u}^{-\beta}|e^{v}_{u}|^{q-2}\,d\langle e^{v}\rangle_{u}} - \beta q\int_{0}^{t}{v_{u}^{-(\beta+1)}|e^{v}_{u}|^{q-1}\sgn(e^{v}_{u})\,d\langle v,e^{v}\rangle_{u}},
\end{align}
where $\sgn(e^{v})=1$ if $e^{v}>0$ and $\sgn(e^{v})=-1$ otherwise, and hence
\begin{align}\label{eq3.1.61}
v_{t}^{-\beta}|e^{v}_{t}|^{q} &= -\beta k\theta\int_{0}^{t}{v_{u}^{-(\beta+1)}|e^{v}_{u}|^{q}\,du} + \beta k\int_{0}^{t}{v_{u}^{-\beta}|e^{v}_{u}|^{q}\,du} - \beta\xi\int_{0}^{t}{v_{u}^{-\left(\beta+\frac{1}{2}\right)}|e^{v}_{u}|^{q}\,dW^{v}_{u}} \nonumber\\[2pt]
&- qk\int_{0}^{t}{v_{u}^{-\beta}|e^{v}_{u}|^{q-1}\sgn(e^{v}_{u})(v_{u}-\bar{v}_{u})\,du} + q\xi\int_{0}^{t}{v_{u}^{-\beta}|e^{v}_{u}|^{q-1}\sgn(e^{v}_{u})\big(\sqrt{v_{u}}-\sqrt{\bar{v}_{u}}\hspace{1pt}\big)\,dW^{v}_{u}} \nonumber\\[2pt]
&+ \frac{1}{2}\hspace{1pt}\beta(\beta+1)\xi^{2}\int_{0}^{t}{v_{u}^{-(\beta+1)}|e^{v}_{u}|^{q}\,du} + \frac{1}{2}\hspace{1pt}q(q-1)\xi^{2}\int_{0}^{t}{v_{u}^{-\beta}|e^{v}_{u}|^{q-2}\big|\sqrt{v_{u}}-\sqrt{\bar{v}_{u}}\hspace{1pt}\big|^{2}\,du} \nonumber\\[2pt]
&- \beta q\xi^{2}\int_{0}^{t}{v_{u}^{-\left(\beta+\frac{1}{2}\right)}|e^{v}_{u}|^{q-1}\sgn(e^{v}_{u})\big(\sqrt{v_{u}}-\sqrt{\bar{v}_{u}}\hspace{1pt}\big)\,du}.
\end{align}
We can show that the two stochastic integrals in \eqref{eq3.1.61} are true martingales by a simple application of H\"older's inequality and Lemmas \ref{Lem3.1}, \ref{Lem3.3} and \ref{Lem3.7}. Taking expectations on both sides, since
\begin{equation}\label{eq3.1.62}
v_{u}-\bar{v}_{u}=e^{v}_{u}+\tilde{v}_{u}-\bar{v}_{u} \hspace{1em}\text{ and }\hspace{1em} \sgn(e^{v}_{u})e^{v}_{u}=|e^{v}_{u}|,
\end{equation}
we deduce that
\begin{align}\label{eq3.1.63}
\E\Big[v_{t}^{-\beta}|e^{v}_{t}|^{q}\Big] &= (\beta-q)k\E\bigg[\int_{0}^{t}{v_{u}^{-\beta}|e^{v}_{u}|^{q}\,du}\bigg] - qk\E\bigg[\int_{0}^{t}{v_{u}^{-\beta}|e^{v}_{u}|^{q-1}\sgn(e^{v}_{u})(\tilde{v}_{u}-\bar{v}_{u})\,du}\bigg] \nonumber\\[2pt]
&- \beta q\xi^{2}\E\bigg[\int_{0}^{t}{v_{u}^{-\left(\beta+\frac{1}{2}\right)}|e^{v}_{u}|^{q-2}e^{v}_{u}\big(\sqrt{v_{u}}-\sqrt{\bar{v}_{u}}\hspace{1pt}\big)\,du}\bigg] \nonumber\\[2pt]
&+ \frac{1}{2}\hspace{1pt}q(q-1)\xi^{2}\E\bigg[\int_{0}^{t}{v_{u}^{-\beta}|e^{v}_{u}|^{q-2}\big|\sqrt{v_{u}}-\sqrt{\bar{v}_{u}}\hspace{1pt}\big|^{2}\,du}\bigg] \nonumber\\[2pt]
&- \frac{1}{2}\hspace{1pt}\beta\xi^{2}(\nu-\beta-1)\E\bigg[\int_{0}^{t}{v_{u}^{-(\beta+1)}|e^{v}_{u}|^{q}\,du}\bigg].
\end{align}
However, note that
\begin{align}\label{eq3.1.64}
\sqrt{v_{u}}\hspace{1pt}e^{v}_{u}\big(\sqrt{v_{u}}-\sqrt{\bar{v}_{u}}\hspace{1pt}\big) &= \sqrt{v_{u}}\hspace{1pt}(v_{u}-\bar{v}_{u}+\bar{v}_{u}-\tilde{v}_{u})\big(\sqrt{v_{u}}-\sqrt{\bar{v}_{u}}\hspace{1pt}\big) \nonumber\\[3pt]
&\geq v_{u}\big|\sqrt{v_{u}}-\sqrt{\bar{v}_{u}}\hspace{1pt}\big|^{2} - \sqrt{v_{u}}\hspace{1pt}\big|\sqrt{v_{u}}-\sqrt{\bar{v}_{u}}\hspace{1pt}\big||\tilde{v}_{u}-\bar{v}_{u}| \nonumber\\[3pt]
&\geq v_{u}\big|\sqrt{v_{u}}-\sqrt{\bar{v}_{u}}\hspace{1pt}\big|^{2} - |e^{v}_{u}||\tilde{v}_{u}-\bar{v}_{u}| - |\tilde{v}_{u}-\bar{v}_{u}|^{2}.
\end{align}
Substituting back into \eqref{eq3.1.63} with \eqref{eq3.1.64} leads to
\begin{align}\label{eq3.1.65}
\E\Big[v_{t}^{-\beta}|e^{v}_{t}|^{q}\Big] &\leq (\beta-q)^{+}k\E\bigg[\int_{0}^{t}{v_{u}^{-\beta}|e^{v}_{u}|^{q}\,du}\bigg] + qk\E\bigg[\int_{0}^{t}{v_{u}^{-\beta}|e^{v}_{u}|^{q-1}|\tilde{v}_{u}-\bar{v}_{u}|\,du}\bigg] \nonumber\\[2pt]
&+ \beta q\xi^{2}\E\bigg[\int_{0}^{t}{v_{u}^{-(\beta+1)}|e^{v}_{u}|^{q-1}|\tilde{v}_{u}-\bar{v}_{u}|\,du}\bigg] + \beta q\xi^{2}\E\bigg[\int_{0}^{t}{v_{u}^{-(\beta+1)}|e^{v}_{u}|^{q-2}|\tilde{v}_{u}-\bar{v}_{u}|^{2}\,du}\bigg] \nonumber\\[2pt]
&- \frac{1}{2}\hspace{1pt}q\xi^{2}(2\beta+1-q)\E\bigg[\int_{0}^{t}{v_{u}^{-\beta}|e^{v}_{u}|^{q-2}\big|\sqrt{v_{u}}-\sqrt{\bar{v}_{u}}\hspace{1pt}\big|^{2}\,du}\bigg]  \nonumber\\[2pt]
&- \frac{1}{2}\hspace{1pt}\beta\xi^{2}(\nu-\beta-1)\E\bigg[\int_{0}^{t}{v_{u}^{-(\beta+1)}|e^{v}_{u}|^{q}\,du}\bigg].
\end{align}
Let $\eta>0$. For any $a,b\geq0$ and $j\in\{1,2\}$, Young's inequality yields
\begin{equation}\label{eq3.1.66}
a^{q-j}b^{j} = \Big(\eta^{\frac{j(q-j)}{q}}a^{q-j}\Big)\Big(\eta^{-\frac{j(q-j)}{q}}b^{j}\Big) \leq \frac{q-j}{q}\hspace{1pt}\eta^{j}a^{q} + \frac{j}{q}\hspace{1pt}\eta^{j-q}b^{q}.
\end{equation}
Using \eqref{eq3.1.66} and after some rearrangements, we deduce that
\begin{align}\label{eq3.1.67}
\E\Big[v_{t}^{-\beta}|e^{v}_{t}|^{q}\Big] &\leq \big((\beta-q)^{+}+\eta(q-1)\big)k\E\bigg[\int_{0}^{t}{v_{u}^{-\beta}|e^{v}_{u}|^{q}\,du}\bigg] + \eta^{1-q}k\E\bigg[\int_{0}^{t}{v_{u}^{-\beta}|\tilde{v}_{u}-\bar{v}_{u}|^{q}\,du}\bigg] \nonumber\\[2pt]
&+ \eta^{1-q}(1+2\eta)\beta\xi^{2}\E\bigg[\int_{0}^{t}{v_{u}^{-(\beta+1)}|\tilde{v}_{u}-\bar{v}_{u}|^{q}\,du}\bigg] \nonumber\\[2pt]
&- \frac{1}{2}\hspace{1pt}q\xi^{2}(2\beta+1-q)\E\bigg[\int_{0}^{t}{v_{u}^{-\beta}|e^{v}_{u}|^{q-2}\big|\sqrt{v_{u}}-\sqrt{\bar{v}_{u}}\hspace{1pt}\big|^{2}\,du}\bigg]  \nonumber\\[2pt]
&- \frac{1}{2}\hspace{1pt}\beta\xi^{2}\big(\nu-\beta-1-2\eta(q-1)-2\eta^{2}(q-2)\big)\E\bigg[\int_{0}^{t}{v_{u}^{-(\beta+1)}|e^{v}_{u}|^{q}\,du}\bigg].
\end{align}
There are two cases to consider. First, if $q\leq2\beta+1$, as we can find $\eta>0$ small enough such that
\begin{equation}\label{eq3.1.68}
2\eta(q-1)+2\eta^{2}(q-2) \leq \beta
\end{equation}
and since $2\beta+1<\nu$ from Lemma \ref{Lem3.7}, we get
\begin{align}\label{eq3.1.69}
\E\Big[v_{t}^{-\beta}|e^{v}_{t}|^{q}\Big] &\leq \big((\beta-q)^{+}+\eta(q-1)\big)k\E\bigg[\int_{0}^{t}{v_{u}^{-\beta}|e^{v}_{u}|^{q}\,du}\bigg] + \eta^{1-q}k\E\bigg[\int_{0}^{t}{v_{u}^{-\beta}|\tilde{v}_{u}-\bar{v}_{u}|^{q}\,du}\bigg] \nonumber\\[2pt]
&+ \eta^{1-q}(1+2\eta)\beta\xi^{2}\E\bigg[\int_{0}^{t}{v_{u}^{-(\beta+1)}|\tilde{v}_{u}-\bar{v}_{u}|^{q}\,du}\bigg].
\end{align}
Second, if $q>2\beta+1$, using the triangle and AM-GM inequalities yields
\begin{equation}\label{eq3.1.70}
v_{u}\big|\sqrt{v_{u}}-\sqrt{\bar{v}_{u}}\hspace{1pt}\big|^{2} \leq |v_{u}-\bar{v}_{u}|^{2} \leq (1+\eta)|e^{v}_{u}|^{2} + \eta^{-1}(1+\eta)|\tilde{v}_{u}-\bar{v}_{u}|^{2},
\end{equation}
and hence, using Young's inequality, we get
\begin{equation}\label{eq3.1.71}
qv_{u}|e^{v}_{u}|^{q-2}\big|\sqrt{v_{u}}-\sqrt{\bar{v}_{u}}\hspace{1pt}\big|^{2} \leq \big(q+\eta(q-2)\big)(1+\eta)|e^{v}_{u}|^{q} + 2\eta^{1-q}(1+\eta)|\tilde{v}_{u}-\bar{v}_{u}|^{q}.
\end{equation}
Substituting back into \eqref{eq3.1.67} with \eqref{eq3.1.71} leads to
\begin{align}\label{eq3.1.72}
\E\Big[v_{t}^{-\beta}|e^{v}_{t}|^{q}\Big] &\leq \big((\beta-q)^{+}+\eta(q-1)\big)k\E\bigg[\int_{0}^{t}{v_{u}^{-\beta}|e^{v}_{u}|^{q}\,du}\bigg] + \eta^{1-q}k\E\bigg[\int_{0}^{t}{v_{u}^{-\beta}|\tilde{v}_{u}-\bar{v}_{u}|^{q}\,du}\bigg] \nonumber\\[2pt]
&+ \eta^{1-q}\xi^{2}\big((1+2\eta)\beta+(1+\eta)(q-2\beta-1)\big)\E\bigg[\int_{0}^{t}{v_{u}^{-(\beta+1)}|\tilde{v}_{u}-\bar{v}_{u}|^{q}\,du}\bigg] \nonumber\\[4pt]
&- \frac{1}{2}\hspace{1pt}\xi^{2}\big(\beta(\nu-\beta-1)-q(q-2\beta-1)-2\eta(q-1)(q-\beta-1)-\eta^{2}(q-1)(q-2)\big) \nonumber\\[3pt]
&\times \E\bigg[\int_{0}^{t}{v_{u}^{-(\beta+1)}|e^{v}_{u}|^{q}\,du}\bigg].
\end{align}
However, note that
\begin{equation}\label{eq3.1.73}
\beta(\nu-\beta-1)-q(q-2\beta-1) > 0
\end{equation}
 from Lemma \ref{Lem3.7} and hence we can find $\eta>0$ small enough such that
\begin{equation}\label{eq3.1.74}
 2\eta(q-1)(q-\beta-1)+\eta^{2}(q-1)(q-2) \leq \beta(\nu-\beta-1)-q(q-2\beta-1).
\end{equation}
Going back to \eqref{eq3.1.72}, we deduce that
\begin{align}\label{eq3.1.75}
\E\Big[v_{t}^{-\beta}|e^{v}_{t}|^{q}\Big] &\leq \big((\beta-q)^{+}+\eta(q-1)\big)k\E\bigg[\int_{0}^{t}{v_{u}^{-\beta}|e^{v}_{u}|^{q}\,du}\bigg] + \eta^{1-q}k\E\bigg[\int_{0}^{t}{v_{u}^{-\beta}|\tilde{v}_{u}-\bar{v}_{u}|^{q}\,du}\bigg] \nonumber\\[2pt]
&+ \eta^{1-q}\xi^{2}\big((1+2\eta)\beta+(1+\eta)(q-2\beta-1)\big)\E\bigg[\int_{0}^{t}{v_{u}^{-(\beta+1)}|\tilde{v}_{u}-\bar{v}_{u}|^{q}\,du}\bigg].
\end{align}
Let $0\leq s\leq T$. Combining \eqref{eq3.1.69} and \eqref{eq3.1.75}, taking the supremum over $[0,s]$ and using Fubini's theorem leads to
\begin{align}\label{eq3.1.76}
\sup_{t\in[0,s]}\E\Big[v_{t}^{-\beta}|e^{v}_{t}|^{q}\Big] &\leq \big((\beta-q)^{+}+\eta(q-1)\big)k\int_{0}^{s}{\sup_{t\in[0,u]}\E\Big[v_{t}^{-\beta}|e^{v}_{t}|^{q}\Big]du} \nonumber\\[2pt]
&+ \eta^{1-q}\xi^{2}\big((1+2\eta)\beta+(1+\eta)(q-2\beta-1)^{+}\big)T\sup_{t\in[0,T]}\E\Big[v_{t}^{-(\beta+1)}|\tilde{v}_{t}-\bar{v}_{t}|^{q}\Big] \nonumber\\[2pt]
&+ \eta^{1-q}kT\sup_{t\in[0,T]}\E\Big[v_{t}^{-\beta}|\tilde{v}_{t}-\bar{v}_{t}|^{q}\Big].
\end{align}
Note from Lemma \ref{Lem3.7} that we can find $r>1$ such that
\begin{equation}\label{eq3.1.77}
\frac{\nu}{\nu-\beta-1} < \frac{r}{r-1} < \frac{2(\nu-\bar{\nu}-1)}{q}\hspace{1pt}.
\end{equation}
Applying H\"older's inequality yields, for $i\in\{0,1\}$,
\begin{equation}\label{eq3.1.78}
\sup_{t\in[0,T]}\E\Big[v_{t}^{-(\beta+i)}|\tilde{v}_{t}-\bar{v}_{t}|^{q}\Big] \leq \sup_{t\in[0,T]}\E\Big[v_{t}^{-r(\beta+i)}\Big]^{\frac{1}{r}}\sup_{t\in[0,T]}\E\Big[|\tilde{v}_{t}-\bar{v}_{t}|^{\frac{rq}{r-1}}\Big]^{\frac{r-1}{r}}.
\end{equation}
Substituting back into \eqref{eq3.1.76} with \eqref{eq3.1.78} and using Gronwall's inequality leads to
\begin{align}\label{eq3.1.79}
\sup_{t\in[0,T]}\E\Big[v_{t}^{-\beta}|e^{v}_{t}|^{q}\Big] &\leq \sup_{t\in[0,T]}\E\Big[|\tilde{v}_{t}-\bar{v}_{t}|^{\frac{rq}{r-1}}\Big]^{\frac{r-1}{r}}e^{((\beta-q)^{+}+\eta(q-1))kT}\bigg\{\eta^{1-q}kT\sup_{t\in[0,T]}\E\Big[v_{t}^{-r\beta}\Big]^{\frac{1}{r}} \nonumber\\[2pt]
&+ \eta^{1-q}\xi^{2}\big((1+2\eta)\beta+(1+\eta)(q-2\beta-1)^{+}\big)T\sup_{t\in[0,T]}\E\Big[v_{t}^{-r(\beta+1)}\Big]^{\frac{1}{r}}\bigg\}.
\end{align}
Using Lemma \ref{Lem3.1}, Proposition \ref{Prop3.6} and \eqref{eq3.1.77}, we conclude that there exist $N_{0}\in\mathbb{N}$ and a constant $C$ such that, for all $N>N_{0}$,
\begin{equation}\label{eq3.1.80}
\sup_{t\in[0,T]}\E\Big[v_{t}^{-\beta}|e^{v}_{t}|^{q}\Big] \leq CN^{-\frac{q}{2}}.
\end{equation}
From H\"older's inequality, Lemma \ref{Lem3.1} and \eqref{eq3.1.80}, we have
\begin{equation}\label{eq3.1.81}
\sup_{t\in[0,T]}\E\Big[|e^{v}_{t}|^{p}\Big] \leq \sup_{t\in[0,T]}\E\Big[v_{t}^{-\beta}|e^{v}_{t}|^{q}\Big]^{\frac{p}{q}}\sup_{t\in[0,T]}\E\Big[v_{t}^{\frac{\beta p}{q-p}}\Big]^{1-\frac{p}{q}} \leq CN^{-\frac{p}{2}}.
\end{equation}
Finally, upon noticing that
\begin{equation}\label{eq3.1.82}
|v_{t}-\bar{v}_{t}| \leq |v_{t}-\tilde{v}_{\bar{t}}| \leq |e^{v}_{t}| + |\Delta\tilde{v}_{t}|,
\end{equation}
the conclusion follows from \eqref{eq3.1.47} and \eqref{eq3.1.81}.
\end{proof}

\section{Numerical results}\label{sec:numerics}

In this section, we perform a numerical analysis of the strong convergence of the FTE scheme. We denote by $\bar{v}_{T,\hspace{1pt}N}$ the value at time $T$ of the approximation process corresponding to an equidistant discretization with $N$ time steps, and study the $L^{1}$ error
\begin{equation}\label{eq4.1}
\varepsilon(N) = \E\big[|v_{T}-\bar{v}_{T,\hspace{1pt}N}|\big].
\end{equation}

\begin{figure}[!htb]
\centering
\begin{subfigure}{.5\textwidth}
  \centering
	\captionsetup{justification=centering}
  \includegraphics[width=.98\linewidth,height=1.68in]{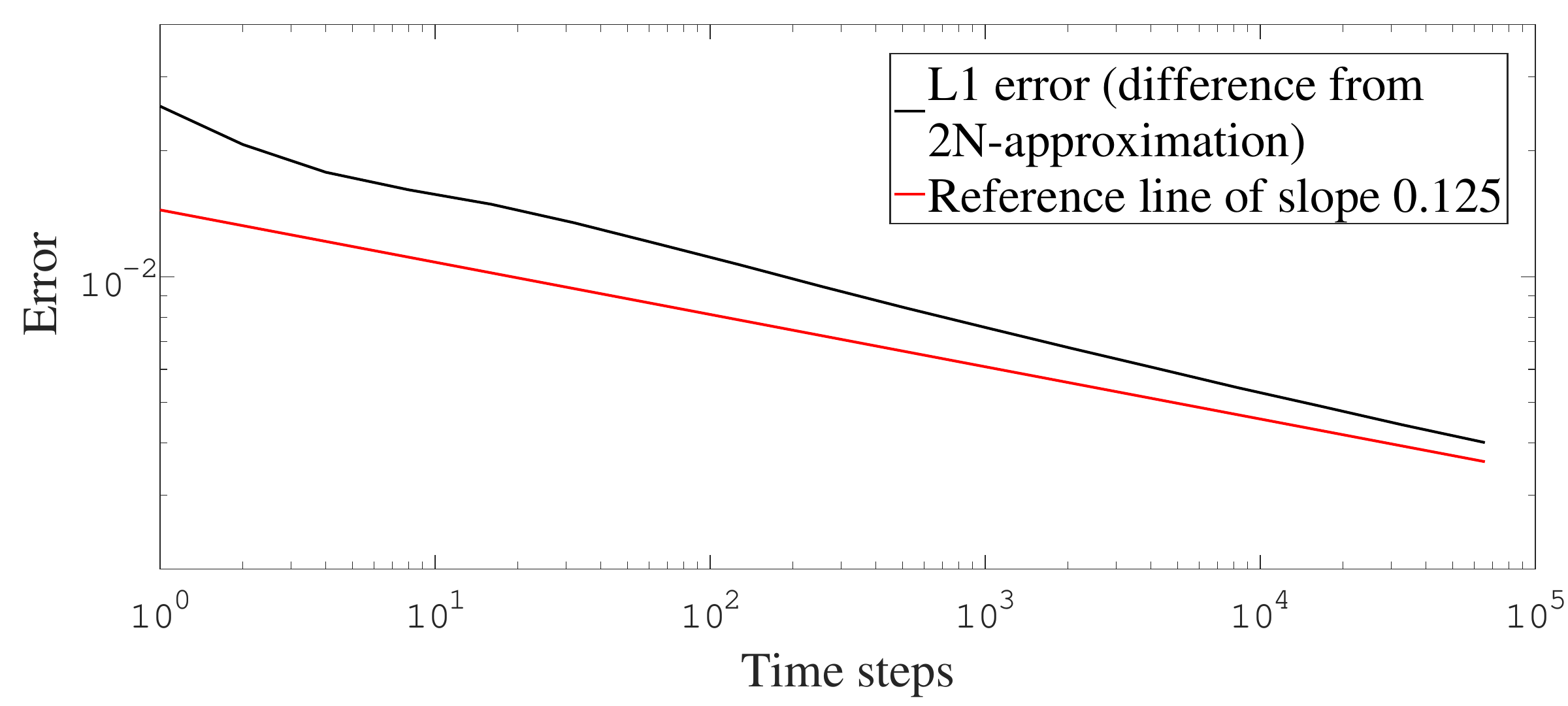}
  \caption{$\nu=0.125$}
  \label{fig:1a}
\end{subfigure}%
\begin{subfigure}{.5\textwidth}
  \centering
	\captionsetup{justification=centering}
  \includegraphics[width=.98\linewidth,height=1.68in]{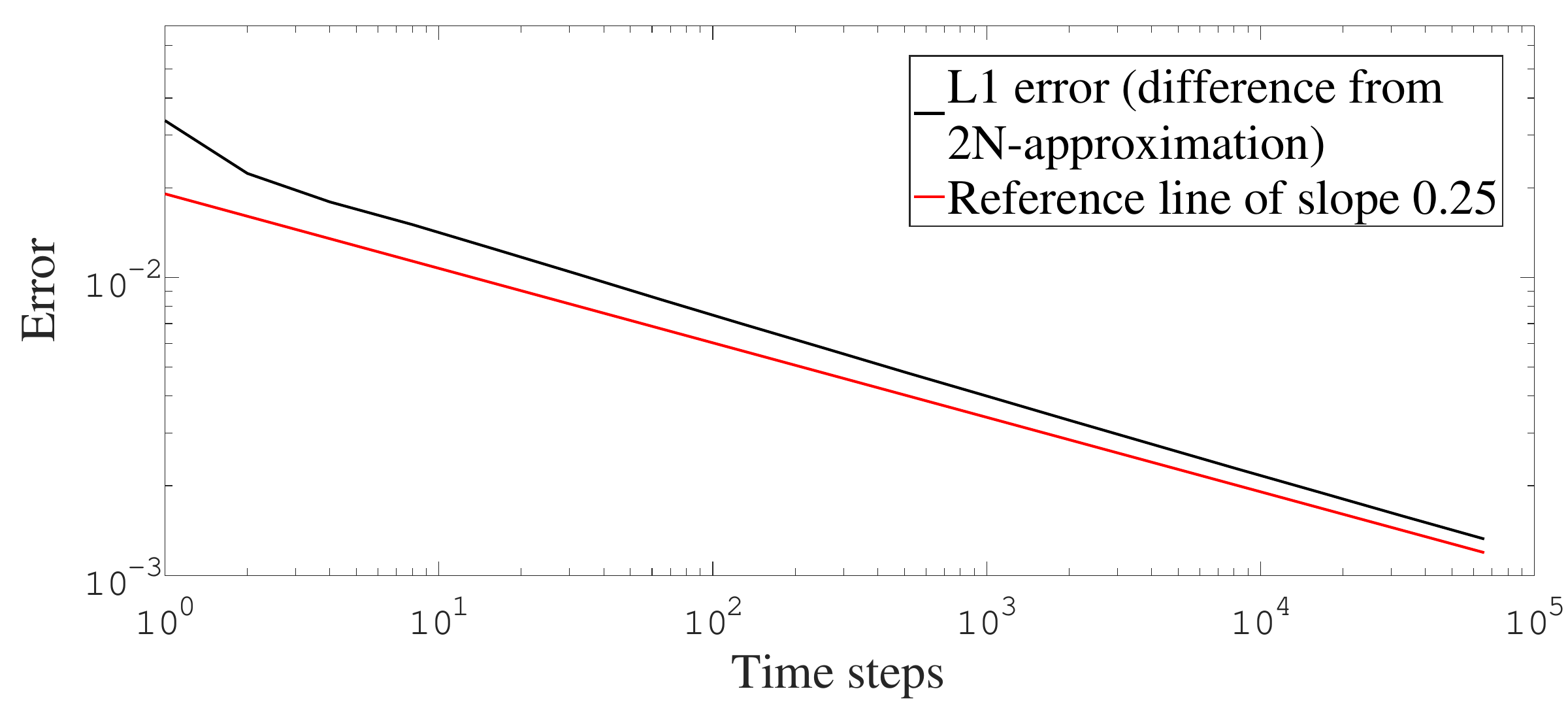}
  \caption{$\nu=0.25$}
  \label{fig:1b}
\end{subfigure} \\[.5em]
\begin{subfigure}{.5\textwidth}
  \centering
	\captionsetup{justification=centering}
  \includegraphics[width=.98\linewidth,height=1.68in]{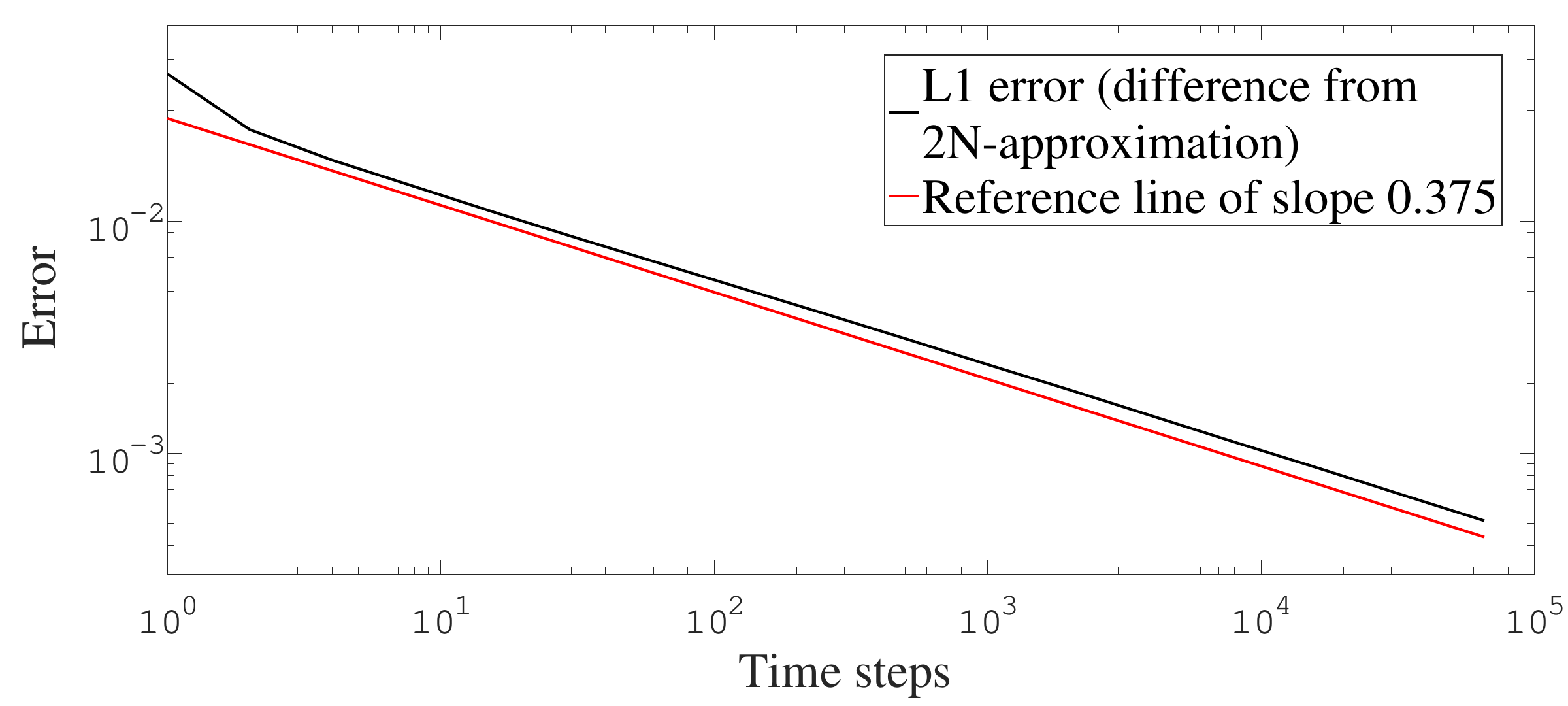}
  \caption{$\nu=0.375$}
  \label{fig:1c}
\end{subfigure}%
\begin{subfigure}{.5\textwidth}
  \centering
	\captionsetup{justification=centering}
  \includegraphics[width=.98\linewidth,height=1.68in]{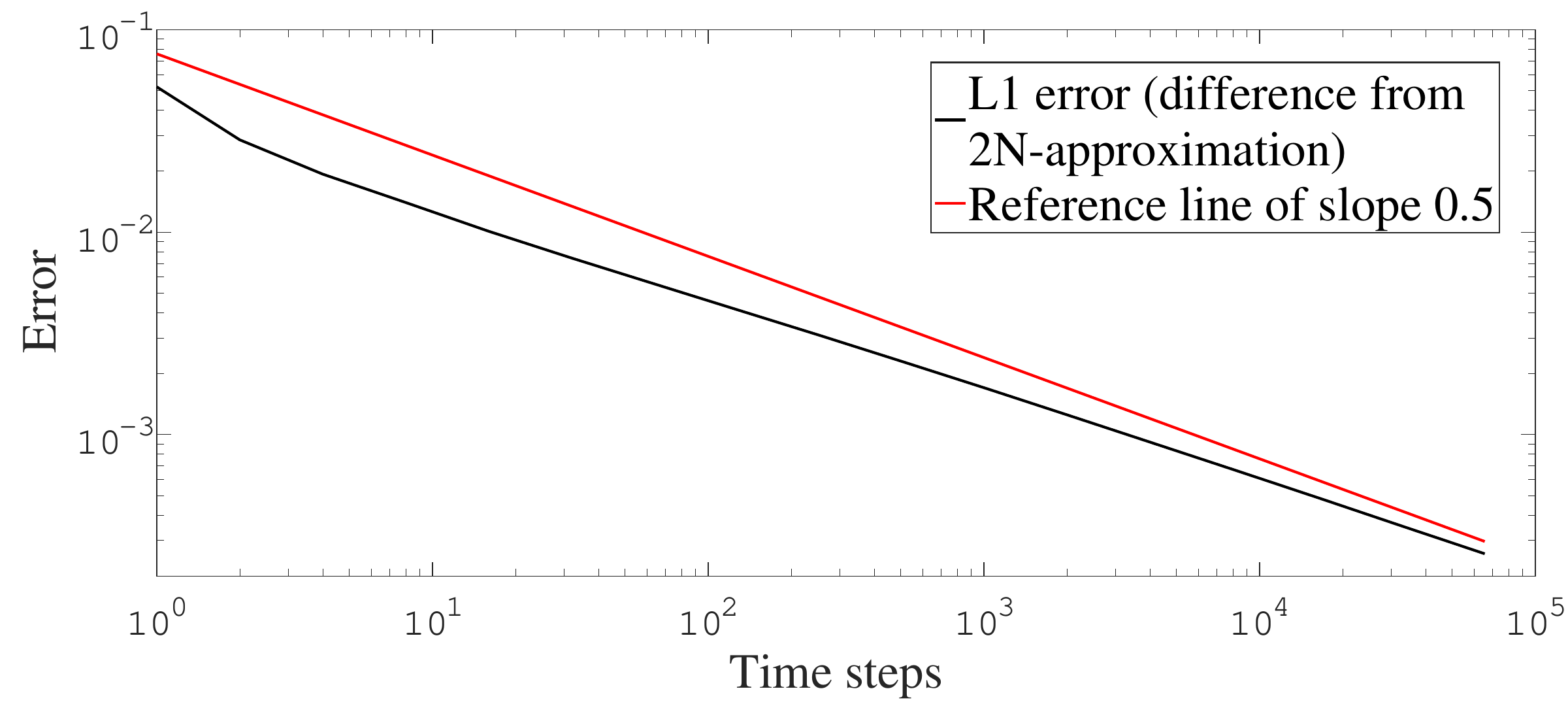}
  \caption{$\nu=0.5$}
  \label{fig:1d}
\end{subfigure} \\[.5em]
\begin{subfigure}{.5\textwidth}
  \centering
	\captionsetup{justification=centering}
  \includegraphics[width=.98\linewidth,height=1.68in]{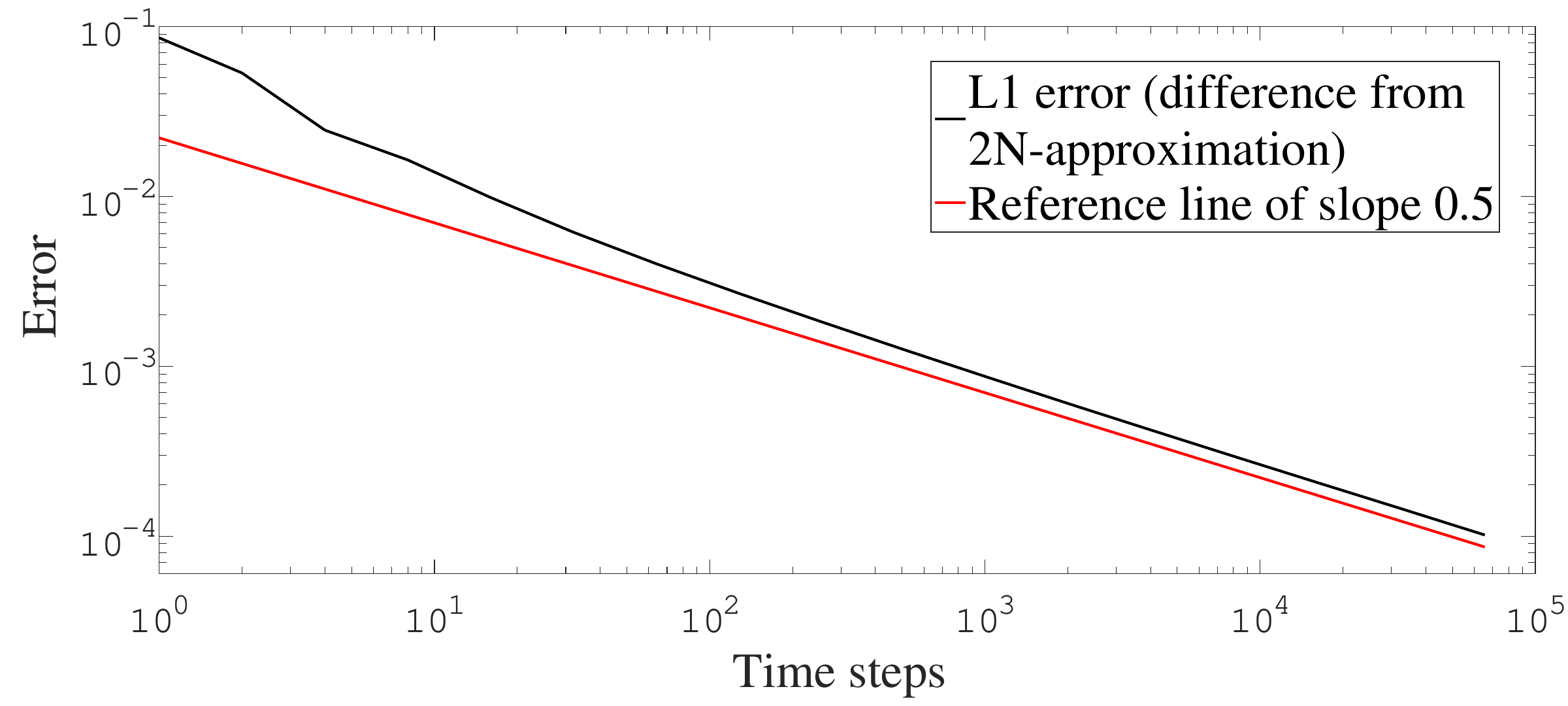}
  \caption{$\nu=1.0$}
  \label{fig:1e}
\end{subfigure}%
\begin{subfigure}{.5\textwidth}
  \centering
	\captionsetup{justification=centering}
  \includegraphics[width=.98\linewidth,height=1.68in]{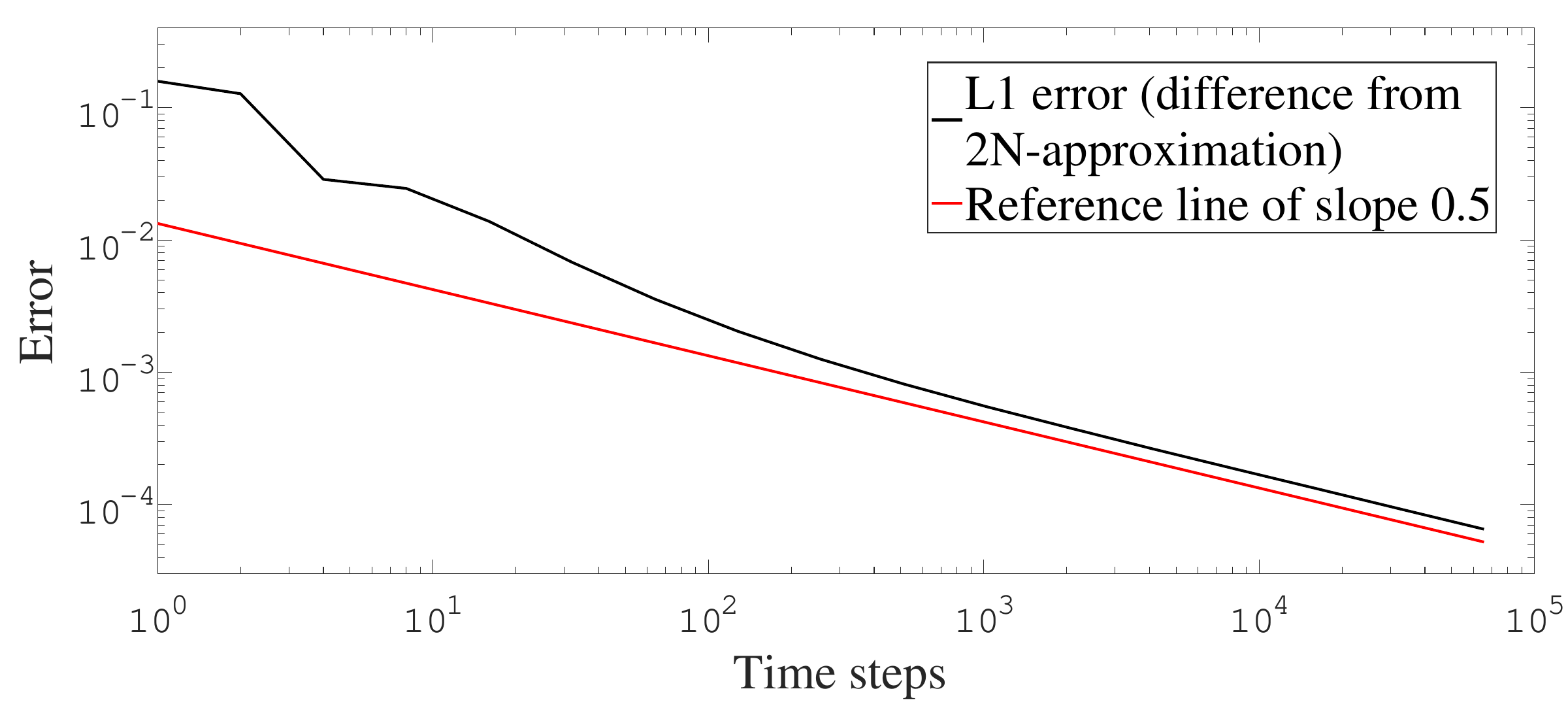}
  \caption{$\nu=2.0$}
  \label{fig:1f}
\end{subfigure} \\[.5em]
\begin{subfigure}{.5\textwidth}
  \centering
	\captionsetup{justification=centering}
  \includegraphics[width=.98\linewidth,height=1.68in]{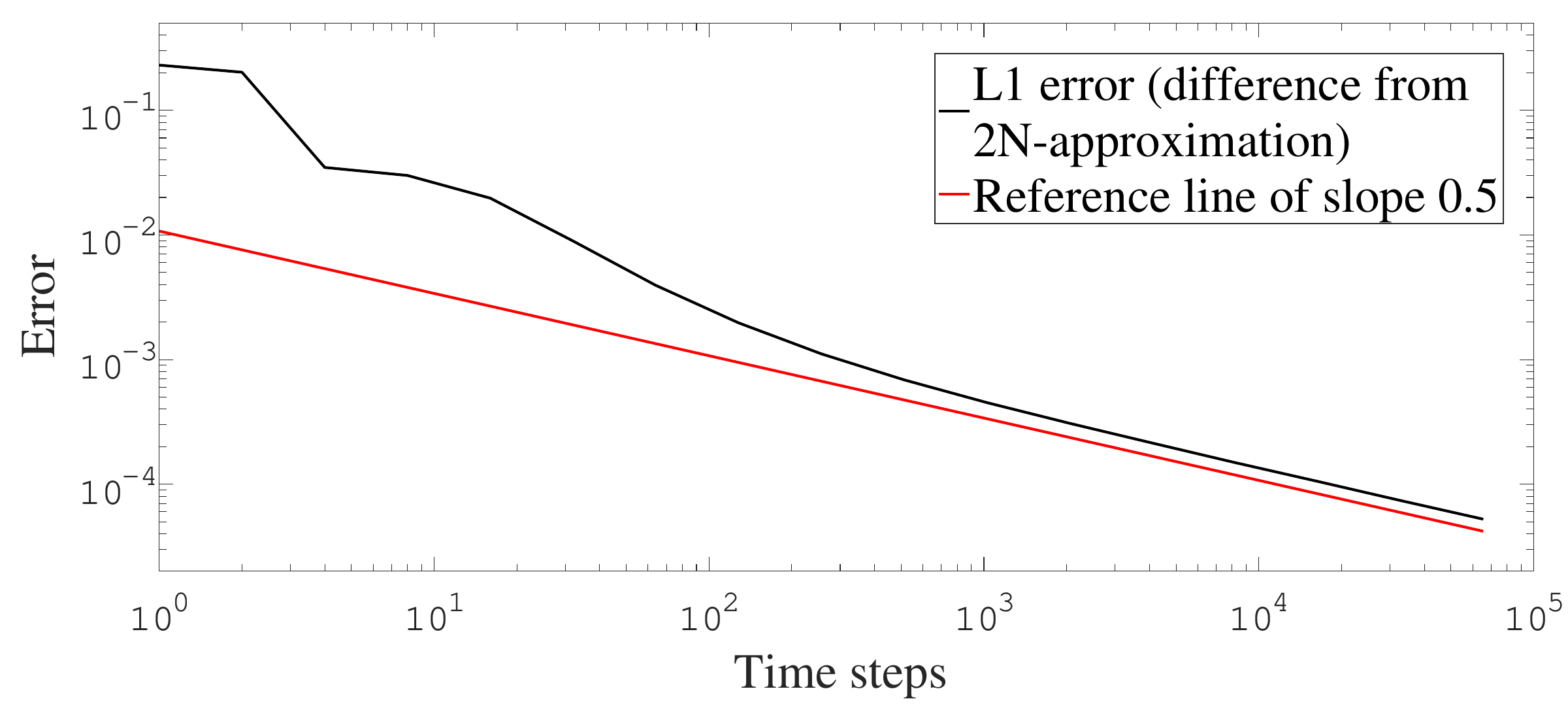}
  \caption{$\nu=3.0$}
  \label{fig:1g}
\end{subfigure}%
\begin{subfigure}{.5\textwidth}
  \centering
	\captionsetup{justification=centering}
  \includegraphics[width=.98\linewidth,height=1.68in]{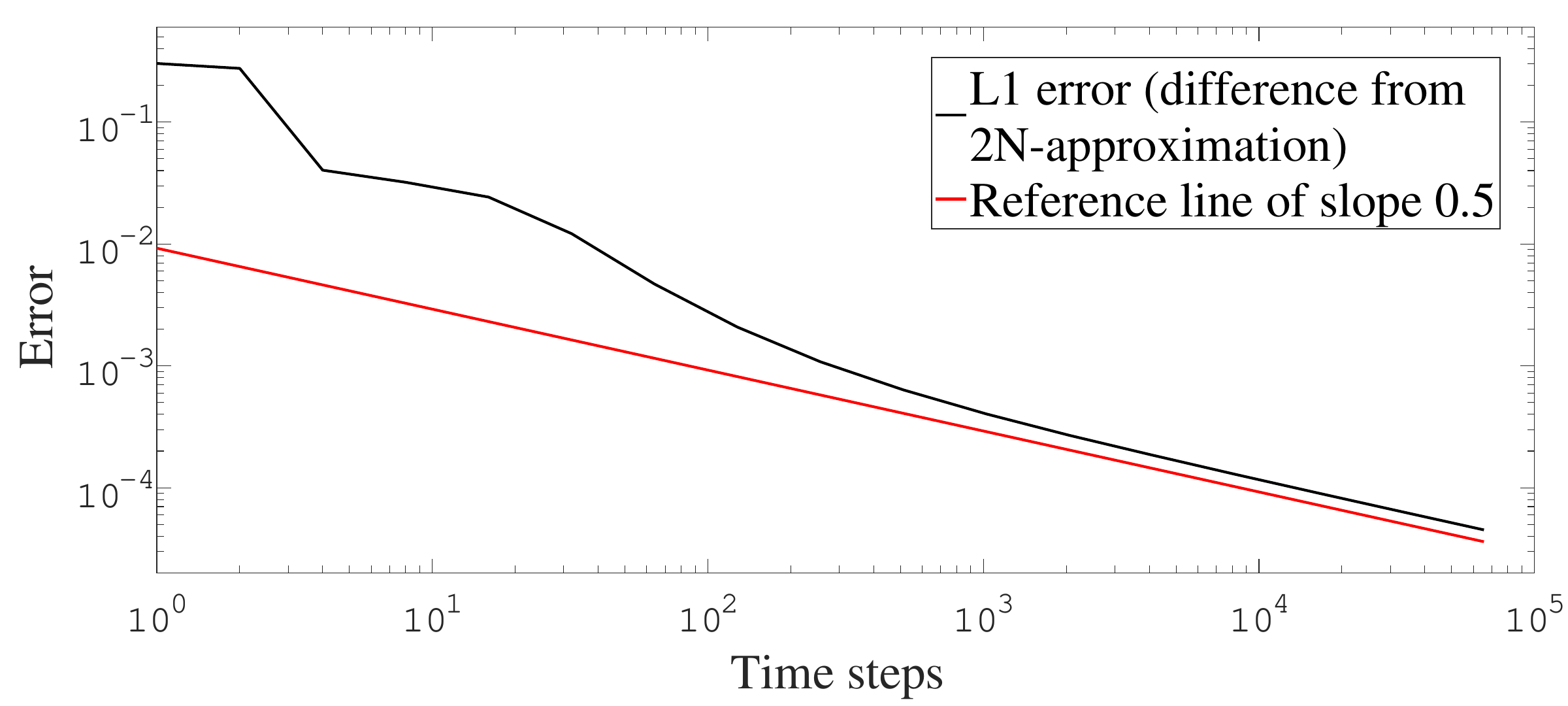}
  \caption{$\nu=4.0$}
  \label{fig:1h}
\end{subfigure}
\caption{The log-plots of the $L^{1}$ errors against the number of time steps when $v_{0}=0.02$, $\theta=0.02$, $k\in\{2.0,4.0,6.0,8.0,16.0,32.0,48.0,64.0\}$, $\xi=0.8$ and $T=1.0$, computed using $2\!\times\!10^{6}$ Monte Carlo paths (for a relative error less than 10 bps).}
\label{fig:1}
\end{figure}

Due to the difficulty in computing this quantity, we estimate as proxy the difference between the values of the approximation process corresponding to $N$ time steps ($\bar{v}_{T\hspace{-1pt},\hspace{1pt}N}$) and $2N$ time steps ($\bar{v}_{T,\hspace{1pt}2N}$) for the same Brownian path, and use the fact that, for any $\alpha>0$,
\begin{equation}\label{eq4.2}
\varepsilon(N) = \mathcal{O}\big(N^{-\alpha}\big) \hspace{1em}\Leftrightarrow\hspace{1em} \E\big[|\bar{v}_{T,\hspace{1pt}N}-\bar{v}_{T,\hspace{1pt}2N}|\big] = \mathcal{O}\big(N^{-\alpha}\big).
\end{equation}
A proof of \eqref{eq4.2} can be found, for instance, in \cite{Alfonsi:2005}. The data in Figure \ref{fig:1} suggest an empirical $L^{1}$ order of $\nu\wedge\frac{1}{2}$, which is in line with our theoretical results when $\nu>3$. We mention that this $L^{1}$ convergence order was demonstrated theoretically (for all $\nu>0$) and numerically (for $\nu=0.25$) for the truncated Milstein scheme \cite{Hefter:2016}. We now recall the main result in \cite{Hefter:2017}, which established a lower error bound for all discretization schemes for the CIR process based on equidistant evaluations of the Brownian driver in the case of an accessible boundary point.
\begin{proposition}[Theorem 1 in \cite{Hefter:2017}]\label{Prop4.1}
If $\nu<1$, then discretization methods for the CIR process $v$ based on equidistant evaluations of the driving noise process achieve at most a strong convergence order of $\nu$, i.e., there exists a positive constant $C$ such that, for all $N\geq1$,
\begin{equation}\label{eq4.3}
\inf_{\varphi\hspace{.5pt}:\hspace{.5pt}\mathbb{R}^{N}\rightarrow\mathbb{R}}\E\Big[\big|v_{T}-\varphi\big(W^{v}_{t_{1}},W^{v}_{t_{2}},\dots,W^{v}_{t_{N}}\big)\big|\Big] \geq CN^{-\nu}.
\end{equation}
\end{proposition}

In particular, the bound in \eqref{eq4.3} suggests an optimal performance -- in the $L^{1}$ sense -- of the FTE scheme in half of the regime where the boundary point is accessible, specifically, when $\nu\leq\frac{1}{2}$.

\section{Conclusions}\label{sec:conclusion}

This work has answered questions concerning the convergence order of the full truncation Euler scheme for the Cox--Ingersoll--Ross process. This scheme is often encountered in the mathematical finance literature in the context of Monte Carlo simulations for multi-dimensional models with Cox--Ingersoll--Ross dynamics in one or more components. One consequence of this work is that we can establish positive strong convergence rates for approximations of these models. For instance, using the log-Euler scheme and the full truncation Euler scheme to discretize the asset price process and the squared volatility in the Heston model, respectively, yields an approximation for which we can easily deduce the strong convergence with order 1/2, by using Theorem \ref{Thm3.8} together with a moment bound result of \cite{Cozma:2016a}. Inevitably, this work also raises some questions, like whether we can relax the condition on the parameters in the inaccessible boundary case without losing the convergence, or whether we can deduce similar properties of the scheme in the accessible boundary case.

\bibliographystyle{abbrv}
\bibliography{references}

\end{document}